\documentclass[a4paper,10pt,reqno]{amsart}

\usepackage{hyperref}
\usepackage{amssymb}
\usepackage{eucal}
\usepackage{eufrak}
\usepackage{graphics}
\usepackage{graphicx}
\usepackage{epsfig}
\usepackage{multicol}
\usepackage{xypic}
\usepackage{amsmath}
\usepackage{wasysym}
\usepackage{dsfont}

\usepackage{eufrak}
\usepackage{graphics}
\usepackage{graphicx}
\usepackage{epsfig}
\usepackage{multicol}
\usepackage{xypic}
\usepackage{amsmath}
\usepackage{color}

\newtheorem{proposition}{Proposition}[section]
\newtheorem{definition}{Definition}[section]

\newtheorem{remark}[definition]{Remark}

\newtheorem{theorem}[definition]{Theorem}

\newcommand{\lp}{\left(}
\newcommand{\rp}{\right)}
\newcommand{\lc}{\left\{}
\newcommand{\rc}{\right\}}
\newcommand{\der}{\partial}

\newcommand{\w}{\wedge}
\newcommand{\bra}{\langle}
\newcommand{\ket}{\rangle}
\newcommand{\R}{\mathds{R}}      

\newcommand{\Flder}{\rightarrow}

\usepackage{amssymb}

\newcommand{\proa}{A^*G \mbox{$\;$}_{\tau^*} \kern-3pt\times_\alpha
G \mbox{$\;$}_\beta \kern-3pt\times_{\tau^*} A^*G}

\newcommand{\al}{\mathfrak{g}}
\newcommand{\dal}{\mathfrak{g}^{*}}

\newcommand{\OmF}{\Omega_F}
\newcommand{\Fcal}{\mathcal{F}}
\newcommand{\Fcalsh}{\mathcal{F}^{\sharp}}
\newcommand{\Fcalshd}{\mathcal{F}_{\sharp}}
\newcommand{\OmTh}{\Omega_{\Theta}}
\newcommand{\Omsh}{\Omega^{\sharp}}
\newcommand{\PSB}{\mathcal{F}^{\sharp}\oplus\mathcal{F}_{\sharp}}
\newcommand{\tx}{\tilde x}
\newcommand{\BF}{\mathbf{B}}
\newcommand{\EF}{\mathbf{E}}
\newcommand{\En}{\mathcal{E}}
\newcommand{\Apo}{\mathbf{A}}

\begin{document}

\title{Hamilton-Dirac systems for charged particles in gauge fields}

\author[F. Jim\'enez]{Fernando Jim\'enez}
\address{F. Jim\'enez: Zentrum Mathematik der Technische Universit\"at M\"unchen, D-85747 Gar\-ching bei M\"unchen, Germany.} \email{fjimenez@ma.tum.de}

\thanks{This research was supported by the DFG Collaborative Research Center TRR 109, ``Discretization in Geometry and Dynamics''. 
}

\maketitle

\begin{abstract}
In this work, we use the Sternberg phase space (which may be considered as the classical phase space of particles in gauge fields) in order to explore the dynamics of such particles in the context of Hamilton-Dirac systems and their associated Hamilton-Pontryagin variational principles. For this, we develop an analogue of the Pontryagin bundle in the case of the Sternberg phase space. Moreover, we show the link of this new bundle to the so-called magnetized Tulczyjew triple, which is an analogue of the link between the Pontryagin bundle and the usual Tulczyjew triple. Taking advantage of the symplectic nature of the Sternberg space, we induce a Dirac structure on the Sternberg-Pontryagin bundle which leads to the Hamilton-Dirac structure that we are looking for. We also analyze the intrinsic and variational nature of the equations of motion of particles in gauge fields in regards of the defined new geometry. Lastly, we illustrate our theory through the case of a $U(1)$ gauge group, leading to the paradigmatic example of an electrically charged particle in an electromagnetic field.
\end{abstract}

\section{Introduction}

In the Hamiltonian formalism, many classical mechanical systems are described by a
manifold, which plays the role of phase space, endowed with a symplectic structure and
a choice of Hamiltonian function. More concretely, if $S$ is a smooth manifold equipped with a symplectic two-form $\Omega_S$, i.e. $(S,\Omega_S)$, the dynamics induced by a smooth Hamiltonian function $H:S\Flder\R$, embodied in its Hamiltonian vector field $X_H:S\Flder TS$, is determined by the well-known Hamiltonian equations
\[
\mathbf{i}_{X_H}\Omega_S=dH.
\]
As can be noticed, these equations are global and may be derived from the pure geometry of the phase space. Particularly, the dynamics of a particle with con\-fi\-gu\-ra\-tion manifold $Q$ is determined by its cotangent bundle $(T^*Q,\Omega_{T^*Q})$, the usual phase space in classical mechanics, and a given Hamiltonian function $H:T^*Q\Flder\R$. The Lagrangian counterpart of mechanics is not as {\it geometrical} as the Hamiltonian side, say the Euler-Lagrange equations for a given Lagrangian function $L:TQ\Flder\R$ cannot be obtained from the geometry of the tangent bundle $TQ$. Nevertheless, both approaches may be described intrinsically under the same framework when one combines the theory of Lagrangian submanifolds (see \cite{Weins0,Weins1}) with the so-called Tulczyjew's triple (see \cite{Tulczy3,Tulczy1,Tulczy2}): namely, both Hamiltonian and Lagrangian dynamics are described by suitable Lagrangian submanifolds of the double vector bundle $TT^*Q$. Roughly speaking, a Lagrangian submanifold is a maximally isotropic submanifold of a given symplectic manifold, while the Tulczyjew triple is the set made out of the double vector bundles $T^*T^*Q$, $TT^*Q$, $T^*TQ$ and two symplectomorphisms among them, say $\alpha_Q$, $\beta_Q$. This is a powerful mechanism and it has been widely applied in modern Geometric Mechanics, from continuous to discrete systems or from unconstrained to variationally constrained (meaning {\it vakonomically} constrained) systems, as can be seen in the recent references \cite{Yo0,Yo1/2,Cedric,Edu,Yo1,Gra}.

Mathematically speaking, in a gauge theory with gauge group $G$ formulated over a manifold $Q$, a {\it gauge field} is a connection of the $G-$principal bundle $P\Flder Q$. The addition of a gauge field into the classical particle dynamics is non-trivial, specially when the group is non-abelian. From a symplectic perspective, the  description of the phase space of a particle on a gauge field was initiated by Sternberg in \cite{Stenberg}, giving rise to the so-called {\it Sternberg phase space} $\Fcalsh$, which is a vector space. This construction follows the initial ideas in \cite{Wong}, where the equations of motion of the particle and the gauge field are obtained taking advantage of  a Poisson approach; further developments on this subject may be found in \cite{Mont, Weinstein}. From the physical point of view, the dynamics of a classical particle in interaction with a gauge field is interesting in few cases, being the paradigmatic one the case of a charged particle evolving in space and coupled to an electromagnetic field. Of course, this instance is important for its own sake, but recently some attention has been put upon the magnetized Kepler problems \cite{Bai,Meng0,Meng1}, kind of systems that fit in the setup presented in this work. On the other hand, it is mandatory to mention that gauge fields acquire crucial importance at a quantum level, for instance in Yang-Mills theories \cite{YaMi} such as the Standard Model of particle physics, which is a quantum field theory where the gauge fields play the role of the intermediate bosons of fundamental interactions (see \cite{Waco} for a theoretical perspective on the Standard Model).

Again at a classical level, to obtain the equations of motion of a charged particle subject to a gauge field is not easy, and usually it is achieved in the physical literature through the so called {\it minimal coupling} procedure (which consists on {\it shifting} the classical momenta by the gauge field). In a more elegant and geometrical way, it has been accomplished in the recent work \cite{Meng} the task of deriving these equations in the context of a generalization of the Tulczyjew triple (called the {\it magnetized} Tulczyjew triple, where the role of the cotangent bundle $T^*Q$ is played by the Sternberg phase space $\Fcalsh$) and the Lagrangian submanifold theory.

Although symplectic manifolds are the appropriate spaces to describe Hamiltonian systems and have great importance in modern mathematics, they are not suitable to describe
{\it all} classical systems. Mechanical systems with symmetries are described by Poisson
structures and systems with constraints are described by closed (but not exact, therefore {\it presymplectic})
two-forms. Systems with both symmetries and constraints are described using Dirac structures,
introduced by Courant in the early 1990s \cite{Courant}.  The original idea was to formulate the dynamics of constrained systems, including constraints
induced from degenerate Lagrangians, as in \cite{Dirac1,Dirac2}. As a matter of fact, Hamiltonian systems can be formulated in the context of Dirac
structures, and their application to electric circuits and mechanical systems with nonholonomic
constraints (namely constraints depending on the configuration and velocity variables which, moreover, are {\it not} integrable) was studied in detail in \cite{shaft}
where
they called the associated Hamiltonian systems with Dirac structures {\it implicit Hamiltonian
systems}. On the other hand, in \cite{YM,YM2} it was explored the Lagrangian side of this framework, developing the notion of {\it implicit Lagrangian system} (or {\it Hamilton-Dirac system}) as a Lagrangian analogue of implicit Hamiltonian systems. In spite of the Lagrangian naming, the dynamics of this systems is still Hamiltonian with respect to a Dirac structure. This kind of structures was designed to account for the link between Dirac structures in the cotangent bundle and a degenerate Lagrangian system with nonholonomic constraints. Moreover, the suitable space to derive their equations of motion in a variational fashion, through the {\it Hamilton-Pontryagin principle}, is the so-called {\it Pontryagin bundle} $TQ\oplus T^*Q$. Besides succeeding in the description of electric circuits and nonholonomic mechanics, the Hamilton-Dirac systems can be also applied to constrained variational dynamics as lately shown in \cite{JiYo}.
\medskip

Prior to the main results, for the sake of completeness we give a comprehensive introduction to the subject, and introduce interesting structures such as the magnetized Tulczyjew triple which allows to obtain the equations of motion under study from a geometrical condition. Then, we follow the introduced ideas  and obtain, employing the already defined Sternberg phase space and magnetized Tulczyjew triple, new geo\-me\-tri\-cal structures providing the dynamics of a charged classical particle subject to a gauge field.  Particularly, we will apply a generalized notion of Hamilton-Dirac systems to such particles.  Our formulation is general, and accounts for a non-abelian Lie group $G$. For this, we construct an analogue of the Pontryagin bundle in the case of the Sternberg phase space, which we will name as the {\it Sternberg-Pontryagin bundle}, and, furthermore, a Dirac structure there, taking advantage of a suitable presymplectic structure. Moreover, we will prove that the Sternberg-Pontryagin bundle is the appropriate space to derive variationally the equations of motion of the Hamilton-Dirac system under consideration. We put emphasis on the local properties of these geometrical structures, performing most of the computations in local coordinates. We enclose our main results in theorem \ref{Teoremaco}. The paper is structured as follows:

\S\ref{SPS} is devoted to introduce the Sternberg phase space and to carefully describe its local expression and associated symplectic two-form. In \S\ref{triples} we describe both the usual Tulczyjew triple and its magnetized version. Moreover, the equations of motion of a charged particle in a gauge field \eqref{EoM} are introduced, while they are put in the context of \cite{Meng} in proposition \ref{MainThe}. \S\ref{DiracSec} accounts for the description of Dirac structures and Hamilton-Dirac systems. We employ the Pontryagin bundle to illustrate the Hamilton-Dirac systems in proposition \ref{EjemploFacil}, result which, despite quite natural, is original to the extent of our knowledge.  In \S\ref{SPBundle} the Sternberg-Pontryagin bundle is defined and its relationship with the magnetized Tulczyjew triple shown; moreover we present the Sternberg-Pontryagin Hamilton-Dirac system. \S\ref{MT} contains our main result, split into the propositions \ref{propoVar}, \ref{propoInt} and \ref{Lag-Dirac}, where the desired equations of motion are obtained in the context of the Sternberg-Pontryagin bundle from variational, intrinsic and Dirac points of view, respectively. Finally, our theory is illustrated in \S\ref{Ex} through the paradigmatic example of an electrically charged particle in an electromagnetic field.
\medskip

Regarding the repeated indices, we will employ Einstein's
summation convention in this paper unless otherwise is stated.

\section{The Sternberg phase space}\label{SPS}

Throughout this work we assume that $Q$ is a smooth manifold, $G$ is a compact connected Lie group with Lie algebra $\al$, $\pi:P\Flder Q$ is a principal $G-$bundle with a fixed principal connection form $\Theta$, and $F$ is a Hamiltonian $G-$space with symplectic form $\OmF$ and  equivariant moment map $\Phi:F\Flder\dal$ (meaning commutative with respect to the $G-$action), where $\dal$ is the dual of the algebra. By Hamiltonian $G-$space we mean that $F$ is a symplectic manifold
with symplectic form $\OmF$, that $G$ acts on $F$ as a group of
symplectomorphisms, so that there is a homomorphism
of the Lie algebra $\al$ into the algebra of Hamiltonian vector
fields, and that we are given a lifting of this homomorphism to
a homomorphism of $\al$ into the Lie algebra of functions on $F$
(where the Lie algebra structure is given by Poisson bracket). Assuming that $Q$ is $n-$dimensional while $F$ is $m-$dimensional, we denote $(q^i)$, $i=1,...,n$, and $(z^{\alpha})$, $\alpha=1,...,m$ (with $m$ an even number since $F$ is a symplectic manifold), as their local coordinates respectively (we will use $(q,z)$ with some abuse of notation). 

Let $\Fcal:=P\times_GF$; the manifold $\Fcalsh$ is a vector bundle over $\Fcal$, making the following diagram commutative:
\begin{equation}\label{SPhaseDiag}
\begin{split}
\xymatrix{
\Fcalsh\ar[r]^{\tilde\pi_Q}\ar[d]_{\rho^{\sharp}} &\Fcal\ar[d]^{\rho}\\
T^*Q\ar[r]_{\pi_Q}&Q
}
\end{split}
\end{equation}
where $\pi_Q$ is the canonical projection. It was proven in \cite{Stenberg} that there is a correct substitute $\OmTh$ on $\Fcal$ for $\OmF$ on $F$, in the sense that it is a closed two-form on $\Fcal$ and is equal to $\OmF$ when $P\Flder Q$ is a trivial bundle with the product connection. Furthermore, if $\Omega_{T^*Q}$ is the canonical symplectic form on $T^*Q$, then
\begin{equation}\label{SterSym}
\Omsh:=\Omega_{T^*Q}+\OmTh
\end{equation}
is a symplectic two-form on $\Fcalsh$ (which we will name henceforth as the Sternberg symplectic form). For sake of simplicity, we shall use the same notation for both the differential form (or a map) and its pullback under a fiber bundle projection map (for instance, in \eqref{SterSym} both the symplectic two-form on $T^*Q$ and its pullback by $\rho^{\sharp}$ are denoted by $\Omega_{T^*Q}$, while $\Omega_{\Theta}$ denotes both a two-form on $\Fcal$ and its pullback through $\tilde\pi_Q$; therefore the sum of both two-forms makes sense).

In  order to describe these elements from the local point of view, we consider a local trivialization $\phi$ of the principal bundle $\pi:P\Flder Q$, namely  local diffeomorphisms $\phi:Q\times G\Flder P$ and $\phi_F:Q\times F\Flder\Fcal$. Then, if $(q,z)$ are local coordinates of $\Fcal$ and $(q,p)$  of $T^*Q$ (where obviously $p$ stands for $p_i$), the commutativity of diagram \eqref{SPhaseDiag} establishes $(q,p,z)$ as local coordinates of $\Fcalsh$ and the following local expression of the projections:
\begin{eqnarray*}
\tilde\pi_Q:\Fcalsh&\Flder&\Fcal;\,\,\,\,\quad\quad\quad\,\,\,\tilde\pi_Q:(q,p,z)\mapsto(q,z),\\
\rho^{\sharp}:\Fcalsh&\Flder& T^*Q;\,\,\,\,\quad\quad\rho^{\sharp}:(q,p,z)\,\,\,\mapsto(q,p),\\
\rho:\Fcal&\Flder& Q; \,\,\,\,\,\quad\quad\quad\,\,\,\rho:\,\,(q,z)\,\,\,\,\,\,\mapsto(q).
\end{eqnarray*}
Needless to say, these local projections stress the nature $\Fcalsh$ as vector bundle.

Regarding the Sternberg symplectic form, we present the needed results, and refer to \cite{Meng,Stenberg} for further details.
\begin{proposition}
There exists a closed real differential well-defined two-form $\OmTh$ on $\Fcal$ defined by $\OmTh:=\Omega_F-d\bra A,\Phi\ket$ under a local trivialization of $P\Flder Q$, where the connection $\Theta$ is represented by the $\al-$valued differential one-form $A$ on $Q$.
\end{proposition}
We point out that $A$ is the local representation of the connection $\Theta$ under trivialization. The uniqueness of $\OmTh$ is proved \cite{Meng}, lemma 2.1.  Finally, the two-form $\Omsh$ defined in \eqref{SterSym} is established as a symplectic form on $\Fcalsh$ through the following proposition:
\begin{proposition}
The differential two-form $\Omsh$ is a symplectic form on $\Fcalsh$.
\end{proposition}
\begin{proof}
It is quite easy to see that $\OmTh$ is closed, since it is composed by the already closed two-form $\Omega_F$ and a total differential. Besides, $\Omega_{T^*Q}$ is closed as a symplectic two-form, making $\Omsh$ also closed.

On the other hand, consider the local coordinates of $\Fcalsh$, $(q,p,z)$, and the local form of $\Omsh$, namely
\begin{equation}\label{OmshLocal}
\Omsh=dq^i\w dp_i+\frac{1}{2}\Omega_{\alpha\beta}\,dz^{\alpha}\w dz^{\beta}+\frac{1}{2}\bra\der_iA_j-\der_jA_i,\Phi\ket\,dq^i\w dq^j-\bra A_i,\der_{\alpha}\Phi\ket\,dq^i\w dz^{\alpha},
\end{equation} 
where $\der_i=\frac{\der}{\der q^i}$, $\der_{\alpha}=\frac{\der}{\der z^{\alpha}}$ and $\Omega_{\alpha\beta}$ is the local expression of the symplectic form $\Omega_F$ on the Hamiltonian space $F$. Employing the matrix form
\begin{equation}\label{Matrix}
\Omsh=\lp\begin{array}{ccc}
\bra\der_iA_j-\der_jA_i,\Phi\ket &\delta^i_j &-\bra A_i,\der_{\alpha}\Phi\ket\\
-\delta^i_j& 0 &0\\
\bra A_i,\der_{\alpha}\Phi\ket&0&\Omega_{\alpha\beta}
\end{array}
\rp
\end{equation}
it is easy to check that $\Omsh$ is non-degenerate everywhere by block reduction. This makes the claim hold.
\end{proof}
The symplectic manifold $(\Fcalsh,\Omsh)$ is referred as the {\it Sternberg phase space}. In \cite{Weinstein} it was introduced a symplectic space out of the principal $G-$bundle $P\Flder Q$ and the Hamiltonian $G-$space $F$, and showed that a connection $\Theta$ yields a symplectomorphism to the Sternberg phase space.
\medskip

Similarly,  $\Fcalshd$ is the dual vector bundle of $\Fcalsh$, making commutative the diagram
\begin{equation}\label{FTQ}
\begin{split}
\xymatrix{
 \Fcalshd\ar[r]^{\tilde\tau_Q}\ar[d]_{\rho_{\sharp}}&\Fcal\ar[d]^{\rho} \\
 TQ\ar[r]_{\tau_Q}&Q,
}
\end{split}
\end{equation}
where $\tau_Q:TQ\Flder Q$ is the canonical projection of the tangent bundle. Introducing local coordinates $(q,v)$ for $TQ$ (where $v$ stands for $v^i$), we may describe locally $\Fcalshd$ by $(q,v,z)$ and the projections in \eqref{FTQ} by:
\begin{eqnarray*}
\tilde\tau_Q:\Fcalshd&\Flder&\Fcal;\,\,\,\,\quad\quad\quad\,\,\,\tilde\tau_Q:(q,v,z)\mapsto(q,z),\\
\rho_{\sharp}:\Fcalshd&\Flder& TQ;\,\,\,\,\quad\quad\rho_{\sharp}:\,\,(q,v,z)\,\,\,\mapsto(q,v).\\
\end{eqnarray*}

\section{The magnetized Tulczyjew triple}\label{triples}
Taking advantage of the symplectic structure of $(\Fcalsh,$ $\Omsh)$ described in \S\ref{SPS} and the relationship between $\Fcalsh$ and $\Fcalshd$, namely they are dual vector bundles of each other, it has been elegantly introduced in \cite{Meng} an analogue of the usual Tulczyjew triple made out of these spaces, named as the {\it magnetized Tulczyjew triple}. We introduce both notions and some other useful results for our purposes.

\subsection{The Tulczyjew triple}
The spaces $TT^{*}Q$, $T^{*}TQ$ and $T^{*}T^{*}Q$ are naturally {\it double vector bundles} (see \cite{Godbillon}, \cite{Pradines}) over $T^{*}Q$ and $TQ$. In \cite{Tulczy1} and \cite{Tulczy2}, Tulczyjew established two symplectomorphisms among these spaces, the first one between $TT^{*}Q$ and $T^{*}TQ$ (namely $\alpha_Q$) and the second one between $TT^{*}Q$ and $T^{*}T^{*}Q$ (namely $\beta_Q$). As cotangent bundles, $T^*TQ$ and $T^*T^*Q$ are naturally equipped with symplectic two-forms, $\Omega_{T^*TQ}$ and $\Omega_{T^*T^*Q}$ respectively. On the other hand, it may be also proven that $TT^*Q$ is a symplectic manifold, equipped with the symplectic two-form $\Omega_{TT^*Q}:=d_T\Omega_{T^*Q}$, where $d_T\Omega_{T^*Q}$ is the tangent lift of $\Omega_{T^*Q}$, which is the usual symplectic form of the cotangent bundle $T^*Q$ (see \cite{Manolo} for more details.) In the following diagram, known as the {\it Tulczyjew triple}, we show the different relationships among these bundles:
\begin{equation}\label{TulczyTriple}
\begin{split}
\xymatrix{
T^{*}T^{*}Q\ar[dr]_{\pi_{T^*Q}}& &TT^{*}Q\ar[dr]_{T\pi_Q}\ar[dl]^{\tau_{T^*Q}}\ar[ll]_{\beta_{Q}}^{\cong}\ar[rr]^{\alpha_{Q}}_{\cong}& &T^{*}TQ\ar@/^-2pc/[llll]_{\kappa_{Q}}^{\cong}\ar[dl]^{\pi_{TQ}}\\
 &T^*Q\ar[dr]_{\pi_Q} & & TQ\ar[dl]^{\tau_Q}&\\
 & &Q & &
}
\end{split}
\end{equation}
where $\kappa_Q:=\beta_Q\circ\alpha_Q^{-1}$. 
\begin{remark}\label{remarkbundles}
{\rm We have introduced the Tulczyjew triple in terms of the canonical symplectic structures corresponding to the double vector bundles $T^*TQ,\,T^*T^*Q$. Ne\-ver\-the\-less, in a more general geometric landscape, one can always establish the isomorphism $T^*E\cong T^*E^*$, for any vector bundle $E\Flder X$, in terms of the canonical pairings \cite{GG}.}
\end{remark}
In order to show the importance of this construction in Geometric Mechanics (and also to describe the procedure employed in the next subsection to obtain geometrically the equations of motion of charged particles in gauge fields), now we briefly discuss how to describe intrinsically both Lagrangian and Hamiltonian mechanics through the Tulczyjew triple, employing as well the notion of Lagrangian submanifold. We use a rather pedestrian definition of the latter concept since a deeper analysis on this subject is not the purpose of this work. Let $(S,\Omega_S)$ be a symplectic manifold and $N\subset S$ a smooth submanifold with inclusion map $\iota$. We say that $N$ is a {\it Lagrangian submanifold} of $S$ if the following conditions hold:
\[
1)\quad \mbox{dim}\,N=\frac{1}{2}\,\mbox{dim}\,S\,\quad \,\mbox{and}\,\quad 2)\quad \iota^*\Omega_S=0.
\]
Consider a Lagrangian (Hamiltonian) function $L:TQ\Flder\R$ ($H:T^*Q\Flder\R$) generating the differential map $dL:TQ\Flder T^*TQ$ ($dH:T^*Q\Flder T^*T^*Q$). It can be proven that  $dL(TQ)\subset T^*TQ$ ($dH(T^*Q)\subset T^*T^*Q$) is a Lagrangian submanifold. Employing $\alpha_Q^{-1}$ ($\beta_Q^{-1}$) we can therefore generate a Lagrangian submanifold of $TT^*Q$ from $dL$ ($dH$), submanifold which determines a system of implicit differential equations whose integrable part can be obtained by applying the integrability constraint algorithm (see \cite{CLMM2003,MMTulczy1995} for more details). Of course, these implicit differential equations represent the Lagrangian dynamics, i.e. they are equivalent to the Euler-Lagrange equations (respectively the Hamiltonian dynamics and the usual Hamiltonian equations).

\subsection{The magnetized Tulczyjew triple}
In the next diagram, in analogy to \eqref{TulczyTriple}, we introduce the {\it magnetized Tulczyjew triple} (see \cite{Meng} for more details): 

\begin{equation}\label{MagTriple}
\begin{split}
\xymatrix{
T^{*}\Fcalsh\ar[dr]_{\pi_{\Fcalsh}}& &T\Fcalsh\ar[dr]_{T_{\Fcal}}\ar[dl]^{\tau_{\Fcalsh}}\ar[ll]_{\beta_{\Fcal}}^{\cong}\ar[rr]^{\alpha_{\Fcal}}_{\cong}& &T^{*}\Fcalshd\ar@/^-2pc/[llll]_{\kappa_{\Fcal}}^{\cong}\ar[dl]^{\pi_{\Fcalshd}}\\
 &\Fcalsh\ar[dr]_{\tilde\pi_Q} & & \Fcalshd\ar[dl]^{\tilde\tau_Q}&\\
 & & \Fcal& &
}
\end{split}
\end{equation}
where the projection $T_{\Fcal}$, for coordinates $(q,p,z,\dot q,\dot p,\dot z)$ of $T\Fcalsh$, is locally defined by $T_{\Fcal}:(q,p,z,\dot q,\dot p,\dot z)\mapsto (q,\dot q,z)$. The symplectic structures on $T^*\Fcalsh$ and $T^*\Fcalshd$ are provided by their cotangent structure ($\Omega_{T^*\Fcalsh}$ and $\Omega_{T^*\Fcalshd}$ respectively), while $\Omega_{T\Fcalsh}$ is defined from $\Omega_{\Theta}$, this is $\Omega_{T\Fcalsh}:=d_T\Omega_{\Theta}$, where again $d_T$  represents the tangent lift.  Thus, the symplectic nature of $\Fcalsh$ allows to establish the isomorphism $\alpha_{\Fcal}$, i.e. $T\Fcalsh \cong T^*\Fcalshd$, while $\kappa_{\Fcal}$, i.e. $T^*\Fcalsh \cong T^*\Fcalshd$, may be established due to the vector bundle nature of $\Fcalsh$, $\Fcalshd$ and their cotangent bundles,  as pointed out in remark \ref{remarkbundles}. Finally, taking advantage again of the vector bundle nature of all the considered spaces, $\beta_{\Fcal}:=\alpha_{\Fcal}\circ\kappa_{\Fcal}$ completes the diagram.

This triple is used in \cite{Meng} in order to obtain the {\it equations of motion of a charged particle in the presence of gauge fields}. For this, in analogy with how the Lagrangian dynamics is obtained employing the usual Tulczyjew triple, a smooth Lagrangian submanifold of $(T\Fcalsh,\Omega_{T\Fcalsh})$ is considered, in particular the submanifold {\it generated} by a Lagrangian function $L_{\sharp}:\Fcalshd\Flder\R$ through the following diagram:
\[
\xymatrix{
\Fcal & & \R\ar[ddrr]^{\tilde c}\ar[ll]_{c}\ar[ddll]^{(c,\dot{\rho\circ c})} &\\\\
\Fcalshd\ar[uu]^{\tilde\tau_Q}\ar[rr]_{dL_{\sharp}}& &T^*\Fcalshd\ar[rr]^{\alpha_{\Fcal}^{-1}}_{\cong} & &T\Fcalsh
}
\]
Let $c:\R\Flder\Fcal$ be a parametrized curve on $\Fcal$ and let $(c(t),\frac{d}{dt}(\rho\circ c(t)))$, where $\rho$ is the projection map defined in diagrams \eqref{SPhaseDiag} and \eqref{FTQ}, be the {\it lifted} curve to $\Fcalshd$ (note that the local coordinates of $(c(t),\frac{d}{dt}(\rho\circ c(t)))$ may be considered, with some abuse of notation, $(q(t),\dot q(t),z(t))$). Now, considering the differential map $dL_{\sharp}:\Fcalshd\Flder T^*\Fcalshd$, we employ the magnetized Tulczyjew triple  \eqref{MagTriple} to obtain the Lagrangian submanifold $\alpha_{\Fcal}^{-1}\lp dL_{\sharp}((c(t),\frac{d}{dt}(\rho\circ c(t))))\rp$.
Finally, taking into account the local expression of $\alpha_{\Fcal}$, say

\[
\begin{split}
\alpha_{\Fcal}:(q,p,z,\dot q,\dot p,\dot z)\mapsto &(q,\dot q,z,\dot p_i-\bra\dot q^j(\der_jA_i-\der_iA_j),\Phi\ket-\bra A_i,\dot z^{\alpha}\der_{\alpha}\Phi\ket, \\& \,\,\,\,p_i,\dot z^{\beta}\Omega_{\beta\alpha}+\bra \dot q^iA_i,\der_{\alpha}\Phi\ket),
\end{split}
\]
and furthermore $\alpha_{\Fcal}^{-1}\lp dL_{\sharp}(\Fcalshd)\rp$, one arrives at the equations of motion:
\begin{equation}\label{EoM}
\begin{split}
\frac{d}{dt} z^{\alpha}&=\Omega^{\alpha\beta}\lp\frac{\der L_{\sharp}}{\der z^{\beta}}-\bra\dot q^kA_k,\,\frac{\der\Phi}{\der z^{\beta}}\ket\rp,\\\\
\frac{d}{dt}\lp\frac{\der L_{\sharp}}{\der\dot q^i}\rp&=\frac{\der L_{\sharp}}{\der q^i}+\dot q^j\bra\frac{\der A_i}{\der q^j}-\frac{\der A_j}{\der q^i},\Phi\ket+\dot z^{\alpha}\bra A_i,\frac{\der\Phi}{\der z^{\alpha}}\ket,
\end{split}
\end{equation}
where $\lp\Omega^{\alpha\beta}\rp=\lp\Omega_{\alpha\beta}\rp^{-1}$ exists, since $\Omega_F$ is full-rank.

\begin{remark}\label{remarc}
{\rm
As it is well known, the Darboux's theorem ensures that, for any point in $F$, there exists an open neighborhood in which the local coordinates $z^{\alpha}$ may be split into $z^{\alpha}=(z^{a},z^{\bar a})$, where $a,\bar a=1,...,m/2$, such that $\Omega_F=\Omega_{\alpha\beta}\,dz^{\alpha}\wedge dz^{\beta}=\delta_{a\bar a}\,dz^{a}\wedge dz^{\bar a}$, where $\delta_{a\bar a}$ is the usual Kronecker delta. Using this particular local representation, the equations \eqref{EoM} read
\begin{equation}\label{EoMDar}
\begin{split}
\frac{d}{dt} z^{a}&=-\delta^{a\bar a}\lp\frac{\der L_{\sharp}}{\der z^{\bar a}}-\bra\dot q^kA_k,\,\frac{\der\Phi}{\der z^{\bar a}}\ket\rp,\\\\
\frac{d}{dt} z^{\bar a}&=\,\,\,\delta^{\bar a a}\lp\frac{\der L_{\sharp}}{\der z^{a}}-\bra\dot q^kA_k,\,\frac{\der\Phi}{\der z^{a}}\ket\rp,\\\\
\frac{d}{dt}\lp\frac{\der L_{\sharp}}{\der\dot q^i}\rp&=\frac{\der L_{\sharp}}{\der q^i}+\dot q^j\bra\frac{\der A_i}{\der q^j}-\frac{\der A_j}{\der q^i},\Phi\ket+\dot z^{\alpha}\bra A_i,\frac{\der\Phi}{\der z^{\alpha}}\ket,
\end{split}
\end{equation}
where, in the last equation, $\alpha=(a,\bar a)$ and $\delta^{a\bar a}$ is the inverse of $\delta_{a\bar a}$. In general, we shall use equally the expressions \eqref{EoM} and \eqref{EoMDar}, preferring the latter in some proofs for convenience.
}
\end{remark}
As shown by this procedure, the equations above may be obtained from a geo\-me\-tri\-cal condition. On the other hand, they can be obtained by usual calculus of variations (as mentioned, but not proved, in \cite{Meng}). We enclose this result in the following proposition, which must be understood as a rephrasing of part of the {\it main theorem} in \cite{Meng}:
\begin{proposition}\label{MainThe}
Let $L_{\sharp}:\Fcalshd\Flder\R$ be a smooth Lagrangian function and $\tilde c:\R\Flder T\Fcalshd$ a smooth curve. For a charged particle with configuration space $Q$, internal space $F$, gauge field $\Theta$ and Lagrangian $L_{\sharp}$, its equations of motion are locally written as \eqref{EoM}-\eqref{EoMDar}, equations that can be obtained from the next two statements (which are equivalent):
\begin{enumerate}
\item $\tilde c(t)\in \alpha_{\Fcal}^{-1}\lp dL_{\sharp}(\Fcalshd))\rp$,
\item[]
\item let $\mathcal{L}_{\sharp}$ be an extended Lagrangian defined by
\begin{equation}\label{ExLag}
\mathcal{L}_{\sharp}:=L_{\sharp}-\bra\dot q^iA_i,\Phi\ket+\Omega_F(z,\dot z),
\end{equation}
where we set $\Omega_F(z,\dot z):=\Omega_F(\dot z^{a}\frac{\der}{\der z^{a}},z^{\bar a}\frac{\der}{\der z^{\bar a}})=\delta_{a\bar a}\dot z^az^{\bar a}$. Then, the stationary condition for the action functional
\[
\int_{t_1}^{t_2}\mathcal{L}_{\sharp}((q(t),z(t),\dot q(t),\dot z(t))\,dt
\]
where the endpoints of $(q(t),z(t))$ are fixed, singles out a curve obeying the equations \eqref{EoM}-\eqref{EoMDar}.
\end{enumerate}
\end{proposition}
Note that the extended Lagrangian $\mathcal{L}_{\sharp}$ is degenerate on $T\Fcalshd$, i.e. if we define the function $\mathcal{L}_{\sharp}:T\Fcalshd\Flder\R$ it is easy to see that $\frac{\der\mathcal{L}_{\sharp}}{\der\dot v}=0$ using the local coordinates $(q,v,z,\dot q,\dot v,\dot z)$ for $T\Fcalshd$. Our task in the subsequent sections, which is the main purpose of this paper, is to reobtain the equations \eqref{EoM}-\eqref{EoMDar} from a new variational principle and a Hamilton-Dirac condition in the Sternberg-Pontryagin bundle, which will be introduced in \S\ref{SPBundle}.

\section{Dirac structures and Hamilton-Dirac systems}\label{DiracSec}

\subsection{Dirac Structures} We first recall the definition of a {\it Dirac structure on a vector space} $V$, say finite dimensional for simplicity (see \cite{Courant} and \cite{CoWe1988}). Let $V^{\ast}$ be the dual space of $V$, and $\langle\cdot \, , \cdot\rangle$
be the natural pairing between $V^{\ast}$ and $V$. Define the
symmetric pairing
$\langle \! \langle\cdot,\cdot \rangle \!  \rangle$
on $V \oplus V^{\ast}$ by
\begin{equation*}
\langle \! \langle\, (v,\alpha),
(\bar{v},\bar{\alpha}) \,\rangle \!  \rangle
=\langle \alpha, \bar{v} \rangle
+\langle \bar{\alpha}, v \rangle,
\end{equation*}
for $(v,\alpha), (\bar{v},\bar{\alpha}) \in V \oplus V^{\ast}$.
A {\it Dirac structure} on $V$ is a subspace $D \subset V \oplus
V^{\ast}$ such that
$D=D^{\perp}$, where $D^{\perp}$ is the orthogonal
of $D$ relative to the pairing
$\langle \! \langle \cdot,\cdot \rangle \!  \rangle$.

Now let $M$ be a smooth manifold and let $TM \oplus T^{\ast}M$ denote the Whitney sum bundle over $M$, namely, the bundle over the base $M$ and with fiber over the point $x \in M $ equal to $T_xM \times T_x^{\ast}M$. In this paper, we shall call a  subbundle $ D_M \subset TM \oplus T^{\ast}M$ a {\it Dirac structure on the manifold $M$}, or  a {\it Dirac structure on the bundle} $\tau_M:TM \to M$, when $D_M(x)$ is a Dirac structure on the vector space $T_{x}M$ at each point $x \in M$. A given two-form $\omega$ on $M$ together with a distribution $\Delta_{M}$ on $M$ determines a Dirac structure on $M$ as follows:
\begin{proposition}\label{DProof}
The two-form $\omega$ determines a Dirac structure $D_M$ on $M$ whose fiber is given for each $x\in M$ as
\begin{equation}\label{DiracManifold}
\begin{split}
D_M(x)=\{ (v_{x}, \alpha_{x}) \in T_{x}M \times T^{\ast}_{x}M
  \; \mid \; & v_{x} \in \Delta_{M}(x), \; \mbox{and} \\ 
  & \alpha_{x}(w_{x})=\omega_{\Delta_{M}}(v_{x},w_{x}) \; \;
\mbox{for all} \; \; w_{x} \in \Delta_{M}(x) \},
\end{split}
\end{equation}
where $\Delta_M\subset TM$ and  $\omega_{\Delta_{M}}$ is the restriction of $\omega$ to $\Delta_{M}$.
\end{proposition}
We refer to \cite{YM} for the proof.

Of course, this proposition is also valid when $\Delta_M=TM\,\, (\omega_{\Delta_M}=\omega)$, which is the case in this work since we do not consider restricted systems, and, furthermore, either for pre-symplectic or symplectic two-forms since the key property to accomplish the result is their skew-symmetry. On the other hand, throughout this work we shall define the Dirac structures in a different but equivalent way to proposition \ref{DProof}. Namely, each two-form $\omega$ on $M$ defines a bundle map $\omega^{\flat}:TM\Flder T^*M$ by
$\omega^{\flat}(v) =\omega(v,\cdot)$.
Consequently, we may equivalently define $D_M(x)$ in \eqref{DiracManifold} as
\begin{equation*}
\begin{split}
D_M(x)=\{ (v_{x}, \alpha_{x}) \in T_{x}M \times T^{\ast}_{x}M
  \; \mid \;  v_{x} \in \Delta_{M}(x), \; \mbox{and} \;
   \alpha_{x}-\omega^{\flat}(x)(v_{x}) \in \Delta^{\circ}_{M}(x) \;
 \},
 \end{split}
\end{equation*}
or in other words $D_M(x):=\mbox{graph}\,\lp\omega^{\flat}\rp\big|_x$.

\subsection{Hamilton-Dirac systems}
As shown just above, the Dirac structures can be given by the graph of the
bundle map associated with the canonical symplectic structure, and hence
it naturally provides a geometric setting for Hamiltonian mechanics. On the other hand, as mentioned in the introduction, the Dirac systems are also useful in the Lagrangian side when one considers degenerate Lagrangian functions and restricted systems \cite{JiYo,YM,YM2}.

Based on the ideas of these references, we next present a rather general definition of a {\it Hamilton-Dirac dynamical system} and its equations of motion; afterwards, we give a significative example.

\begin{definition}\label{Lag-Dirac-Sys}
Consider a Dirac structure $D_M$ on $M$, a curve $x: \R\Flder M$ and the exterior differential $d\gamma:M\Flder T^*M$, where $\gamma:M\Flder\R$ is a smooth function. We define the {\rm Hamilton-Dirac dynamical system} induced by the Dirac structure $D_M$ and the curve $\gamma$ as the pair $(D_M,\gamma)$. Its equations of motion	are given by 
\[
(\dot x(t)\,,\,d\gamma(x(t)))\in D_M(x(t)).
\]
Any curve $x(t)\subset M$, $t_1\leq t\leq t_2$ satisfying this condition is called a {\rm solution curve} of the Hamilton-Dirac system.
\end{definition}

\begin{remark}
{\rm
The systems introduced in the last definition are also called} implicit Lagrangian systems {\rm or} Lagrange-Dirac systems {\rm in references as \cite{JiYo, YM,YM2} in order to emphasize that, in the cases treated, the function $\gamma$ is a Lagrangian function or a generalized Energy function, as also is the case in this paper. However, we state the} Hamiltonian {\rm naming, since the defined dynamics is Hamiltonian with respect to a Dirac structure.}
\end{remark}

We illustrate the Hamilton-Dirac systems by means of the Pontryagin bundle $TQ\oplus T^*Q$ over a manifold $Q$, that is the Whitney sum of the tangent bundle and the cotangent bundle over $Q$, whose fiber at $q\in Q$ is the product $T_qQ\times T_q^*Q$. The Pontryagin bundle is locally described by $(q,v,p)$, and these three projections are naturally defined:
\begin{eqnarray*}
\mbox{pr}_{TQ}&:&TQ\oplus T^*Q\Flder TQ;\,\,\,\,\,\,\,(q,v,p)\mapsto(q,v),\\
\mbox{pr}_{T^*Q}&:&TQ\oplus T^*Q\Flder T^*Q;\,\,\,\,(q,v,p)\mapsto(q,p),\\
\mbox{pr}_{Q}&:&TQ\oplus T^*Q\Flder Q;\,\,\,\,\,\,\,\,\,\,\,(q,v,p)\mapsto(q).
\end{eqnarray*}
The Pontryagin bundle and its projections {\it fits} in the Tulczyjew triple \eqref{TulczyTriple} as in the next diagram:
\begin{small}
\[
\xymatrix{
T^{*}T^{*}Q\ar[dr]_{\pi_{T^*Q}}& &&TT^{*}Q\ar[drr]^{T\pi_Q}\ar[dll]_{\tau_{T^*Q}}\ar[lll]_{\beta_{Q}}^{\cong}\ar[rrr]^{\alpha_{Q}}_{\cong}&& &T^{*}TQ\ar@/^-2pc/[llllll]_{\kappa_{Q}}^{\cong}\ar[dl]^{\pi_{TQ}}\\
 &T^*Q\ar[ddrr]_{\pi_Q}& &TQ\oplus T^*Q\ar[ll]_{\mbox{pr}_{T^*Q}}\ar[rr]^{\mbox{pr}_{TQ}}\ar[dd]^{\mbox{pr}_{Q}}& & TQ\ar[ddll]^{\tau_Q}&\\\\
 & &&Q && &
}
\]
\end{small}
Consider now the presymplectic two-form $\Omega_{T^*Q}$ on $T^*Q\oplus TQ$ (where we denote by $\Omega_{T^*Q}$ its pullback under the projection $\mbox{pr}_{T^*Q}$). Thus, employing the proposition \eqref{DProof}, we can define the Dirac structure
\begin{align*}
D_{PB}(y)
& =\{ (v_{y}, \alpha_{y}) \in T_{y}(TQ\oplus T^{\ast}Q) \times
T^{\ast}_{y}(TQ\oplus T^{\ast}Q)  \mid v_{y} \in
T_{y}(TQ\oplus T^{\ast}Q),\\ \mbox{and} \nonumber
 & \qquad \qquad
\alpha_{y}(w_{y}) = \Omega_{T^*Q}(y) (v_{y},w_{y}) \;\; \mbox{for
all} \;\; w_{y} \in T_{y}(TQ\oplus T^{\ast}Q)\},
\end{align*}
where $y=(q,v,p)\in TQ\oplus T^*Q$, or in the simpler form $D_{PB}(y)=\mbox{graph}\,\lp\Omega_{T^*Q}\rp^{\flat}\big|_y$. Given a Lagrangian $L:TQ\Flder\R$ (possibly degenerate) and its associated generalized energy $E_L:TQ\oplus T^*Q\Flder\R$, $E_L:=\bra p,v\ket-L(q,v)$, according to definition \eqref{Lag-Dirac-Sys} we can state the following proposition: 

\begin{proposition}\label{EjemploFacil}
The equations of motion of the Hamilton-Dirac system $(D_{PB},E_L)$ are locally given for each $y=(q,v,p)\in TQ\oplus T^*Q$ by
\begin{equation}\label{LD1}
\left( (\dot{q},\dot v,\dot{p}), d E_L(q,v,p)\right)  \in D_{PB}(q,v,p).
\end{equation}
These equations are equivalent to the usual Euler-Lagrange equations.
\end{proposition}
\begin{proof}
The Dirac structure $D_{PB}\subset T(TQ\oplus T^*Q)\oplus T^*(TQ\oplus T^*Q)$ is locally defined by
\[
D_{PB}(y)=\lc\lp(\dot q,\dot v, \dot p),(\alpha,\beta,u)\rp\,|\,-\dot p=\alpha,\,0=\beta,\,\dot q=u\rc,
\]
where $\alpha_idq^i+\beta_idv^i+u^idp_i\in T^*(TQ\oplus T^*Q).$ Setting $(\alpha,\beta,u)=dE_L$, we arrive at
$\alpha=-\frac{\partial L}{\partial q}$, $\beta=p -\frac{\partial L}{\partial v }$ and $u=v$, and, therefore, at the coordinate equations of motion of the Hamilton-Dirac system
\[
\dot{p}= \frac{\partial L}{\partial q},\quad p -\frac{\partial L}{\partial v }=0, \quad \dot{q} =v,
\]
which are, after a straightforward computation, the usual Euler-Lagrange equations of a Lagrangian system, namely 
\[
\frac{d}{dt}\lp\frac{\der L}{\der\dot q}\rp=\frac{\der L}{\der q}.
\]
\end{proof}
From the variational point of view, it is easy to prove, employing usual calculus of variations, that these equations can be also obtained from the stationary condition of the action functional
\[
\int_{t_1}^{t_2}\left[\bra p(t),\dot q(t)\ket-E_L(q(t),v(t),p(t))\right]\,dt
\]
with fixed endpoints of $q(t)$. This is known as the {\it Hamilton-Pontryagin principle}.

\section{The Sternberg-Pontryagin bundle and Sternberg-Pontryagin Hamilton-Dirac system}\label{SPBundle}

We use the spaces $\Fcalsh$ and $\Fcalshd$ defined in \S\ref{SPS} in order to introduce, in analogy to the usual Pontryagin bundle $TQ\oplus T^*Q$, what we define as the {\it Sternberg-Pontryagin bundle}.
\begin{definition}
Consider the bundle $\PSB$ over $\Fcal$, whose fiber at $(q,z)\in\Fcal$ is the product $\Fcalsh\times_{(q,z)}{\Fcalshd}$. We call the bundle $\PSB$ the {\rm Sternberg-Pontryagin bundle}.
\end{definition}
Under this definition, the local coordinates of $\PSB$ are written 
\[
(q,v,p,z),
\]
while the following three projections are naturally defined:
\begin{eqnarray*}
\mbox{pr}_{\Fcalsh}&:&\PSB\Flder\Fcalsh;\,\,\,\,(q,v,p,z)\mapsto(q,p,z),\\
\mbox{pr}_{\Fcalshd}&:&\PSB\Flder\Fcalshd;\,\,\,\,(q,v,p,z)\mapsto(q,v,z),\\
\mbox{pr}_{\Fcal}&:&\PSB\Flder\Fcal;\,\,\,\,\,\,(q,v,p,z)\mapsto(q,z).
\end{eqnarray*}
All the previous developments may be summarized into the following diagram, where \eqref{SPhaseDiag} and \eqref{FTQ} have been taken into account and, also, we show how the Sternberg-Pontryagin bundle $\PSB$ {\it fits} in the magnetized Tulczyjew triple \eqref{TulczyTriple}:
\begin{small}
\begin{equation}\label{Diag2}
\begin{split}
\xymatrix{
 T^*\Fcalsh\ar[dr]_{\pi_{\Fcalsh}}\ar@/^+2pc/[rrrrrr]^{\kappa_{\Fcal}} & & & T\Fcalsh\ar[dll]^{\tau_{\Fcalsh}}\ar[drr]_{T_{\Fcal}}\ar[lll]_{\beta_{\Fcal}}\ar[rrr]^{\alpha_{\Fcal}} & & &T^*\Fcalshd\ar[dl]^{\pi_{\Fcalshd}}\\
         & \Fcalsh\ar[ddd]_{\rho^{\sharp}}\ar[ddrr]_{\tilde\pi_Q}& &\PSB\ar[dd]^{\mbox{pr}_{\Fcal}}\ar[rr]_{\mbox{pr}_{\Fcalshd}}\ar[ll]^{\mbox{pr}_{\Fcalsh}} & & \Fcalshd\ar[ddll]^{\tilde\tau_Q}\ar[ddd]^{\rho_{\sharp}}&\\\\
         & & & \Fcal\ar[d]^{\rho}& & & \\
         &T^*Q\ar[rr]_{\pi_Q} & &    Q  & &TQ\ar[ll]^{\tau_Q} &
}
\end{split}
\end{equation}
\end{small}
Taking advantage of the projection $\mbox{pr}_{\Fcalsh}:\PSB\Flder\Fcalsh$, we can induce a presymplectic two-form in the Pontryagin-Sternberg bundle $\PSB$, namely $\lp\mbox{pr}_{\Fcalsh}\rp^*\Omsh$ (which we will also denote $\Omsh$). Furthermore, this two-form induces the bundle map
\[
\lp\Omsh\rp^{\flat}:T(\PSB)\Flder T^*(\PSB),
\]
and consequently, according to proposition \ref{DProof}, the Dirac structure
\begin{equation*}
\begin{split}
D^{\sharp}(x)=\mbox{graph}\,\lp\Omsh\rp^{\flat}\big|_x,
\end{split}
\end{equation*}
where $x=(q,v,p,z)\in\PSB$. We name $D^{\sharp}$ the {\it Pontryagin-Sternberg Dirac structure}. On the other hand, consider a Lagrangian function (possibly degenerate) $L_{\sharp}:\Fcalshd\Flder\R$ and define its associated generalized Energy function $E_{L_{\sharp}}:\PSB\Flder\R$ in local coordinates by
\begin{equation}\label{GenEner}
E_{L_{\sharp}}(q,v,p,z):=\bra p,\,v\ket-L_{\sharp}(q,v,z),
\end{equation}
where $\bra\cdot,\cdot\ket$ denotes the natural pairing between $TQ$ and $T^*Q$. With all these ingredients and according to definition \ref{Lag-Dirac-Sys} we introduce the following Hamilton-Dirac system:
\begin{definition}\label{Hamilton-Dirac-System}
Consider the Dirac structure $D^{\sharp}$ on $\PSB$, a Lagrangian function (possibly degenerate) $L_{\sharp}:\Fcalshd\Flder\R$, its associated generalized Energy function $E_{L_{\sharp}}:\PSB\Flder\R$ \eqref{GenEner} and a curve $x(t)=(q(t),v(t),p(t),z(t))\in\PSB$. We define the {\rm Pontryagin-Sternberg Hamilton-Dirac system} by $(D^{\sharp},E_{L_{\sharp}})$ and its equations of motion by
\begin{equation}
\lp\dot x(t),dE_{L_{\sharp}}(x(t))\rp\in D^{\sharp}(x(t)).
\end{equation}
\end{definition}

\section{Main Theorem}\label{MT}
In this section we split our main result into three propositions, enclosing them in a compact way in the final theorem. The two statements in proposition \eqref{MainThe} might be also included (since they are all equivalent) but we prefer to keep them out in order to emphasize the new results. 
\medskip

First we establish a variational principle providing the equations of motion of a charged particle in a gauge field \eqref{EoM}. For that, we present some useful definitions. As above, let $x=(q,v,p,z)$ be local coordinates of $\PSB$; therefore $(x,\dot x)=(q,v,p,z,\dot q,\dot v,\dot p,\dot z)$ are the local coordinates of $T(\PSB)$. Furthermore, consider $r=(q,v,z)$ local coordinates for $\Fcalshd$ and therefore $(r,\dot r)=(q,v,z,\dot q,\dot v,\dot z)$ for $T\Fcalshd$. Define the {\it extended generalized Energy function} $E_{\mathcal{L}_{\sharp}}:T(\PSB)\Flder\R$, locally given by
\begin{equation}\label{ExEnergy}
E_{\mathcal{L}_{\sharp}}(x,\dot x):=\bra p\,,\,v\ket-\mathcal{L}_{\sharp}(r,\dot r),
\end{equation}
 where $\mathcal{L}_{\sharp}$ is the extended Lagrangian defined in \eqref{ExLag}. Note that $E_{\mathcal{L}_{\sharp}}$ is also degenerate by definition due to its $(\dot v,\dot p)-$independence.
\begin{proposition}\label{propoVar}
Let $\mathcal{L}_{\sharp}:T\Fcalshd\Flder\R$ be a degenerate Lagrangian function defined by \eqref{ExLag} and $E_{\mathcal{L}_{\sharp}}:T(\PSB)\Flder\R$ the degenerate extended generalized energy in \eqref{ExEnergy}. Define the action functional
\begin{equation}\label{AcFunc}
\begin{split}
&\int_{t_1}^{t_2}\left[\bra p(t),\dot q(t)-v(t)\ket+\mathcal{L}_{\sharp}(r(t),\dot r(t))\right]dt\\
&=\int_{t_1}^{t_2}\left[\bra p(t),\dot q(t)\ket-E_{\mathcal{L}_{\sharp}}(x(t),\dot x(t))\right]\,dt.
\end{split}
\end{equation}
Then, keeping the endpoints of $(q(t),z(t))\in\Fcal$ fixed, whereas the endpoints of $v(t)$ and $p(t)$ are allowed to be free, the stationary condition for this action functional induces the equations \eqref{EoM}-\eqref{EoMDar}.
\end{proposition}
\begin{proof}
By direct computations, the variation of \eqref{AcFunc} reads
\[
\begin{split}
\delta\int_{t_1}^{t_2}\left[\bra p,\dot q-v\ket+\mathcal{L}_{\sharp}(r,\dot r)\right]dt&=\int_{t_1}^{t_2}\left[\bra\delta p,\dot q\ket+\bra p,\delta\dot q\ket-\bra\delta p,v\ket-\bra p,\delta v\ket\right.\\&\left. 
+\left<\frac{\der L_{\sharp}}{\der q},\delta q\right>+\left<\frac{\der L_{\sharp}}{\der v},\delta v\right>+\left<\frac{\der L_{\sharp}}{\der z},\delta z\right>\right.\\&\left.
-\left<\delta\dot q^iA_i,\Phi\right>-\left<\dot q^i\der_jA_i\delta q^i,\Phi\right>-\left<\dot q^iA_i,\der_{\alpha}\Phi\delta z^{\alpha}\right>\right.\\&\left.
+\delta_{a\bar a}\,z^{\bar a}\,\delta \dot z^{a}+\delta_{a\bar a}\,\dot z^{a}\,\delta z^{\bar a}\right]dt,
\end{split}
\]
where the particular form of \eqref{ExLag}, i.e. $\mathcal{L}_{\sharp}=L_{\sharp}-\bra\dot q^iA_i,\Phi\ket+\Omega_F(z,\dot z)$, has been taken into account (note in the last two terms the difference between the Kronecker's delta and the variation of the coordinates) as long with the splitting of coordinates $\alpha=(a,\bar a)$. Moreover, in the first four terms  $\bra\cdot,\cdot\ket$ means the pairing between $T^*Q$ and $T^*Q$, in the next three ones the pairing between $T^*\Fcalshd$ and $T\Fcalshd$ and, finally, in the next three ones the pairing between $\al^*$ and $\al$. Now, reordering the terms and performing integration by parts we arrive at
\[
\begin{split}
\delta\int_{t_1}^{t_2}\left[\bra p,\dot q-v\ket+\mathcal{L}_{\sharp}(r,\dot r)\right]dt&=\int_{t_1}^{t_2}\left[ \left<\delta p,\dot q-v\right>+\left<\frac{\der L_{\sharp}}{\der v}-p,\delta v\right>\right.\\&\left.
+\left<-\dot p_i+\frac{\der L_{\sharp}}{\der q^i}+\dot q^j\bra\frac{\der A_i}{\der q^j}-\frac{\der A_j}{\der q^i},\Phi\ket+\bra A_i,\dot z^{\alpha}\der_{\alpha}\Phi\ket\,,\,\delta q^i\right>\right.\\
&\left.+\left<\frac{\der L_{\sharp}}{\der z^{a}}-\bra\dot q^iA_i,\der_{a}\Phi\ket-\delta_{a\bar a}\dot z^{\bar a}\,,\,\delta z^{a}\right>\right]dt\\
&\left.+\left<\frac{\der L_{\sharp}}{\der z^{\bar a}}-\bra\dot q^iA_i,\der_{\bar a}\Phi\ket+\delta_{a\bar a}\dot z^{a}\,,\,\delta z^{\bar a}\right>\right]dt\\
&+\bra p,\delta q\ket\big|_{t_1}^{t_2}-\delta q^i\bra A_i,\Phi\ket\big|_{t_1}^{t_2}+\delta_{a\bar a}\,\delta z^{a}\,z^{\bar a}\big|_{t_1}^{t_2},
\end{split}
\]
where we have used that 
\[
\begin{split}
\int_{t_1}^{t_2}\delta\dot q^i\,\bra A_i,\Phi\ket\,dt&=\delta q^i\bra A_i,\Phi\ket\big|_{t_1}^{t_2}-\int_{t_1}^{t_2}\delta q^i\frac{d}{dt}\bra A_ i,\Phi\ket\,dt\\
&=\delta q^i\bra A_i,\Phi\ket\big|_{t_1}^{t_2}-\int_{t_1}^{t_2}\delta q^i\,\left[\bra\dot q^j\der_jA_ i,\Phi\ket+\bra A_ i,\dot z^{\alpha}\der_{\alpha}\Phi\ket\right]\,dt
\end{split}
\]
under integration by parts. 

Now, taking into account that $\delta q(t_1)=\delta q(t_2)=\delta z(t_1)=\delta z(t_2)=0$ the last three terms vanish. Moreover, considering that $(\delta q,\delta v,\delta p,\delta z)$ are free, the stationary condition above provides the following equations.
\begin{equation}\label{eqconp}
\begin{split}
\dot q=&v,\\
p=&\frac{\der L_{\sharp}}{\der v},\\
\dot p_i=&\frac{\der L_{\sharp}}{\der q^i}+\dot q^j\bra\frac{\der A_i}{\der q^j}-\frac{\der A_j}{\der q^i},\Phi\ket+\bra A_i,\dot z^{\alpha}\der_{\alpha}\Phi\ket,\\
\dot z^{a}=&-\delta^{a\bar a}\lp\frac{\der L_{\sharp}}{\der z^{\bar a}}-\bra\dot q^iA_i,\der_{\bar a}\Phi\ket\rp,\\
\dot z^{\bar a}=&\,\,\,\,\,\delta^{\bar a a}\lp\frac{\der L_{\sharp}}{\der z^{a}}-\bra\dot q^iA_i,\der_{a}\Phi\ket\rp.
\end{split}
\end{equation}
These are obviously the equations \eqref{EoMDar} as claimed.
\end{proof}

Now, taking advantage of the geometry introduced in the diagram  \eqref{Diag2}, we attempt to obtain an intrinsic expressions of the action functional \eqref{AcFunc} and the equations \eqref{EoM}. As a first guess, considering \eqref{Diag2} we notice that the Poincar\'e-Cartan one form $\Theta_{T^*Q}$ on $T^*Q$ (with local form $\Theta_{T^*Q}=p_i dq^i$) can be pulled-back to $T(\PSB)$ through the chain
\[
\xymatrix{
T(\PSB)\ar[r]^{\tau_{\PSB}} &\PSB\ar[r]^{\mbox{pr}_{\Fcalsh}} & \Fcalsh\ar[r]^{\rho^{\sharp}} & T^*Q,
}
\]
where $\tau_{\PSB}:T(\PSB)\Flder\PSB$ is obviously the canonical tangent projection, inducing a one-form $\lp\rho^{\sharp}\circ\mbox{pr}_{\Fcalsh}\circ\tau_{\PSB}\rp^*\,\Theta_{T^*Q}$ on $T(\PSB)$ (denoted $\Theta_{T^*Q}$ as well). Denoting $\tx=(x,\dot x)\in T(\PSB)$, the action functional
\[
\int_{t_1}^{t_2}\left[\bra\Theta_{T^*Q}(\tx(t)),\dot{\tx}(t)\ket-E_{\mathcal{L}_{\sharp}}(\tx(t))\right]\,dt
\]
is a fair global expression of \eqref{AcFunc}, fact that can be easily proven by direct computations in coordinates. Nevertheless, taking variations and integrating by parts we arrive at 
\[
\int_{t_1}^{t_2}\left[\bra-\mathbf{i}_{\dot{\tx}(t)}d\Theta_{T^*Q}(\tx(t))-dE_{\mathcal{L}_{\sharp}}(\tx(t))\,,\,\delta\,{\tx}(t)\ket\right]\,dt+\bra\Theta_{T^*Q}(\tx(t)),\delta\,{\tx}(t)\ket\big|_{t_1}^{t_2}=0,
\]
which fixing the endpoints of $q(t)$ yields $\mathbf{i}_{\dot{\tx}(t)}\Omega_{T^*Q}(\tx(t))=dE_{\mathcal{L}_{\sharp}}(\tx(t))$, with $\Omega_{T^*Q}=-d\Theta_{T^*Q}$. After some calculations, we realize that this is {\it not} a global representation of \eqref{EoM} (we skip the details for sake of short). This fact points out that the usual symplectic geometry, pulled-back to the new space $\PSB$, is not enough to describe the equations of a charged particle in a gauge field. Consequently, we reorient our attention to the Sternberg-Pontryagin bundle in order to construct a meaningful one-form there. Indeed, noting that the connection $A$ is a $\dal-$valued one-form on $T^*Q$ and appealing to the considerations in remark \ref{remarc} we define (using the Darboux's coordinates $(z^a,z^{\bar a})$ for $F$):
\begin{equation}\label{One-Form}
\Theta^{\sharp}:=\lp p_i-\bra A_i,\Phi\ket\rp\,dq^i+z^{\bar a}dz^a,
\end{equation} 
where the one-form in $F$, i.e. $z^{\bar a}dz^a=:\Theta_F$, is defined such that $-d\Theta_F=\Omega_F=\delta_{a\bar a}dz^a\wedge dz^{\bar a}$. Taking into account that $\Omega^{\sharp}=\Omega_{T^*Q}-d\bra A,\Phi\ket+\Omega_F$, it is easy to check that 
$\Omega^{\sharp}=-d\Theta^{\sharp}$. Pulling-back $\Theta^{\sharp}$ to $\PSB$ through $\mbox{pr}_{\Fcalsh}$ (note that $\Omsh$ will be presymplectic in $\PSB$) and taking into account the generalized energy \eqref{GenEner} we define the action functional
\begin{equation}\label{AcFuncCorr}
\int_{t_1}^{t_2}\left[\bra\Theta^{\sharp}(x(t))\,,\,\dot x(t)\ket-E_{L_{\sharp}}(x(t))\right]\,dt,
\end{equation}
where again $x=(q,v,p,z)\in\PSB$, which is as well a fair global expression of \eqref{AcFunc}, as can be easily proven by direct computations in coordinates. We show in the next proposition that this new action functional provides also a global representation of \eqref{EoM}.
\begin{proposition}\label{propoInt}
Under the endpoints {\rm $(q(t),z(t))=\mbox{pr}_{\Fcal}(x(t))$} fixed, the stationary condition of the action functional \eqref{AcFuncCorr} singles out a critical curve $x(t)$ that satisfies the {\rm intrinsic equations of motion of a charged particle in a gauge field}:
\[
\mathbf{i}_{\dot x(t)}\Omega^{\sharp}(x(t))=dE_{L_{\sharp}}(x(t)).
\]
Moreover, these equations are equivalent to \eqref{EoM}.
\end{proposition}
\begin{proof}
To prove the first statement, we take variations over \eqref{AcFuncCorr}, which yields:
\[
\begin{split}
&\delta\int_{t_1}^{t_2}\left[\bra\Theta^{\sharp}(x(t))\,,\,\dot x(t)\ket-E_{L_{\sharp}}(x(t))\right]\,dt\\
&=\int_{t_1}^{t_2}\left[\bra-\mathbf{i}_{\dot{x}(t)}d\Theta^{\sharp}(x(t))-dE_{L_{\sharp}}(x(t))\,,\,\delta\,{x}(t)\ket\right]\,dt+\bra\Theta^{\sharp}(x(t)),\delta\,{x}(t)\ket\big|_{t_1}^{t_2}=0,
\end{split}
\]
where integration by parts has been performed. For all variations $\delta x(t)$ and fixed endpoints $(q(t),z(t))=\mbox{pr}_{\Fcal}(x(t))$, one arrives straightforwardly at $\mathbf{i}_{\dot x(t)}\Omega^{\sharp}(x(t))=dE_{L_{\sharp}}(x(t))$.

To prove the second, we consider the local form of $\Omsh$ on $\PSB$, particularly (recall \eqref{Matrix})
\begin{equation}\label{OmegaSharp}
\Omsh=\lp\begin{array}{cccc}
\bra\der_iA_j-\der_jA_i,\Phi\ket& 0&\delta^i_j &-\bra A_i,\der_{\alpha}\Phi\ket\\
 0&0 &0 &0 \\
-\delta^i_j&0 & 0 &0\\
\bra A_i,\der_{\alpha}\Phi\ket&0 &0&\Omega_{\alpha\beta}
\end{array}
\rp,
\end{equation}
which leads to
\[
\begin{split}
\mathbf{i}_{\dot x}\Omega^{\sharp}(x)=&\left<\dot q^j\bra\der_iA_j-\der_jA_i,\Phi\ket-\dot p_i+\dot z^{\alpha}\bra A_i,\der_{\alpha}\Phi\ket,\,dq^i\right>\\
&+\left<dp,\dot q\right>+\left<-\dot q^i\bra A_i,\der_{\alpha}\Phi\ket+\Omega_{\alpha\beta}\dot z^{\beta},dz^{\alpha}\right>.
\end{split}
\]
On the other hand
\begin{equation}\label{dELsharp}
\begin{split}
dE_{L_{\sharp}}=&\left<\frac{\der E_{L_{\sharp}}}{\der q},dq\right>+\left<\frac{\der E_{L_{\sharp}}}{\der v},dv\right>+\left<dp,\frac{\der E_{L_{\sharp}}}{\der p}\right>+\left<\frac{\der E_{L_{\sharp}}}{\der z},dz\right>\\
=&\left<-\frac{\der L_{\sharp}}{\der q},dq\right>+\left<p-\frac{\der L_{\sharp}}{\der v},dv\right>+\left<dp,v\right>+\left<-\frac{\der L_{\sharp}}{\der z},dz\right>
\end{split}
\end{equation}
Equating both expressions we arrive at equations \eqref{eqconp}, and therefore the claim holds.
\end{proof}
\begin{remark}\label{MinimalCoupling}
{\rm
Roughly speaking, in the definition of the one-form $\Theta^{\sharp}$ we have performed a sort of {\it minimal coupling} condition: namely we have established the substitution $p_i\Flder p_i-\bra A_i,\Phi\ket$, where $p_i$ are the coordinates of the momentum in $T^*Q$. The minimal coupling is the standard procedure in the physics literature to derive the Lorentz equations in a relativistically invariant manner. More concretely, the substitution $p\Flder p-eA$ is made in the  Hamiltonian function (where $p$ is the four-momentum and $A$ is a four-potential of the electromagnetic field, while $e$ is the electric charge). As observed in \cite{Sni,So}, this procedure is equivalent to leaving the Hamiltonian invariant and adding $e\,dA$ to the symplectic form in the original phase space. This is the beginning point by Sternberg himself when constructing the Sternberg's phase space in \cite{Stenberg}.
}
\end{remark}
By means of this proposition we have proven that the suitable space to intrinsically describe the equations of motion of a charged particle in a gauge field is the Pontryagin-Sternberg bundle $\PSB$.

Finally, we employ the Pontryagin-Sternberg Hamilton-Dirac system to reobtain \eqref{EoM}.
\begin{proposition}\label{Lag-Dirac}
Consider the Pontryagin-Sternberg Hamilton-Dirac system \\ $(D^{\sharp},E_{L_{\sharp}})$ defined in \ref{Hamilton-Dirac-System}. Its equations of motion, namely
\[
\lp\dot x(t),dE_{L_{\sharp}}(x(t))\rp\in D^{\sharp}(x(t)),
\]
are equivalent to \eqref{EoM}.
\end{proposition}
\begin{proof}
To prove this, we provide the local expression of $D^{\sharp}=\mbox{graph}\,\lp\Omsh\rp^{\flat}$, which is obtained by considering the local form of $\Omsh$ \eqref{OmegaSharp}. Namely
\[
\begin{split}
D^{\sharp}(x)=\{\lp(\dot q,\dot v, \dot p,\dot z),(\alpha,\beta,u,\mu)\rp\,|&\,\dot q^j\bra\der_iA_j-\der_jA_i,\Phi\ket-\dot p_i+\dot z^{\alpha}\bra A_i,\der_{\alpha}\Phi\ket=\alpha_i,\\ 
&0=\beta_i,\,\dot q^i=u^i,\,-\dot q^i\bra A_i,\der_{\alpha}\Phi\ket+\Omega_{\alpha\beta}\dot z^{\beta}=\mu_{\alpha}\},
\end{split}
\]
where $\alpha_idq^i+\beta_idv^i+u^idp_i+\mu_{\alpha}dz^{\alpha}\in T^*(\PSB)$. When we set $(\alpha,\beta,u,\mu)=dE_{L_{\sharp}}$, which is accomplished by taking into account the local expression \eqref{dELsharp}, we obtain the equations of motion of the Pontryagin-Sternberg Hamilton-Dirac, equations which are obviously equivalent to \eqref{eqconp}, as claimed.
\end{proof}
\begin{remark}
{\rm
The definition of the extended Lagrangian $\mathcal{L}_{\sharp}$ \eqref{ExLag} is crucial in propositions \ref{propoVar} and \ref{propoInt}, where we construct the variational principle and its intrinsic expression in $\PSB$. Despite the particular form of $\mathcal{L}_{\sharp}$ is highly influenced by the Sternberg symplectic structure and therefore {\it quite} natural  (note that $\mathcal{L}_{\sharp}=L_{\sharp}+\bra\Theta^{\sharp},(\dot q,\dot z)\ket-\bra\Theta_{T^*Q},\dot q\ket$, where $\Theta^{\sharp}$ is defined in \eqref{One-Form}), it is completely unnecessary from the Hamilton-Dirac point of view. In fact, we only need $L_{\sharp}$ in order to construct the generalized energy $E_{L_{\sharp}}$, function which forms the Hamilton-Dirac system $(D^{\sharp},E_{L_{\sharp}})$. The Sternberg symplectic structure is only {\it present} in the definition of the Dirac structure $D^{\sharp}$, and consequently in the dynamical condition $\lp\dot x,dE_{L_{\sharp}}(x)\rp\in D^{\sharp}(x)$. In other words, the symplectic structure influences the geometry of the space under study, but it does {\it not} influence its dynamical function, following somehow the Sternberg's {\it program} sketched in remark \ref{MinimalCoupling}.
}
\end{remark}
\medskip

We enclose the results obtained in this section in our main theorem:
\begin{theorem}\label{Teoremaco}
The following statements are equivalent:
\begin{enumerate}
\item The Sternberg-Hamilton-Pontryagin principle for the following action integral
\[
\int_{t_1}^{t_2}\left[\bra p(t),\dot q(t)\ket-E_{\mathcal{L}_{\sharp}}(x(t),\dot x(t))\right]\,dt,
\]
holds for $(q(t),z(t))$ with fixed endpoints.

\item The curve $x(t)=(q(t),v(t),p(t),z(t))\in\PSB$, $t\in[t_1,t_2]$, satisfies the implicit equations
\[
\mathbf{i}_{\dot x(t)}\Omega^{\sharp}(x(t))=dE_{L_{\sharp}}(x(t)),
\]
whose local expression is
\[
\begin{split}
\dot q=&v,\\
p=&\frac{\der L_{\sharp}}{\der v},\\
\dot p_i=&\frac{\der L_{\sharp}}{\der q^i}+\dot q^j\bra\frac{\der A_i}{\der q^j}-\frac{\der A_j}{\der q^i},\Phi\ket+\bra A_i,\dot z^{\alpha}\der_{\alpha}\Phi\ket,\\
\dot z^{\alpha}=&\Omega^{\alpha\beta}\lp\frac{\der L_{\sharp}}{\der z^{\beta}}-\bra\dot q^iA_i,\der_{\beta}\Phi\ket\rp.
\end{split}
\]
\item The curve $x(t)=(q(t),v(t),p(t),z(t))\in\PSB$, $t\in[t_1,t_2]$, is a solution of the Pontryagin-Sternberg Hamilton-Dirac system $(D^{\sharp},E_{L_{\sharp}})$, whose equations of motion are
\[
\lp\dot x(t),dE_{L_{\sharp}}(x(t))\rp\in D^{\sharp}(x(t)).
\]
\end{enumerate}
\end{theorem}

\section{Example}\label{Ex}

As mentioned in the introduction, the paradigmatic example in classical physics of a charged particle subject to a gauge field is an electric charged particle evolving in space and coupled to an electromagnetic field (other interesting examples as the Wong's equations or the magnetized Kepler problems may be found in \cite{Wong} and \cite{Meng0} respectively). We shall consider the autonomous case, i.e. the electromagnetic field does not depend on time, and denote $\EF:=\lc E^i\rc$, $\BF:=\lc B^i\rc$, using the vector notation of physics literature, the electric and magnetic fields, respectively, in the three space coordinates corresponding to $Q=\R^3$ (with local coordinates $\lc q^i\rc=\lc x,y,z\rc$). The textbook equations of motion of a charged particle (charge=$e$ and unit mass $m=1$) coupled to an electromagnetic field $(\EF,\BF)$ are:
\begin{subequations}\label{EqEx}
\begin{align}
\ddot{\mathbf{q}}=&\,\,e\,\left[\EF+\frac{\dot{\mathbf{q}}}{c}\times\BF\right],\label{EqExa}\\
\frac{d}{dt}\mathcal{E}=&\,\,e\,\dot{\mathbf{q}}\cdot\EF,\label{EqExb}
\end{align}
\end{subequations}
where $\mathbf{q}:=\lc q^i\rc$, $\times$ denotes the curl operation , $\cdot$ the scalar product in $\R^3$, $\En$ the energy of the particle and $c$ is the speed of light. Moreover, as it is well-known, both fields may be obtained from the so-called scalar and vector potentials, $\varphi$ and $\Apo$ respectively, by
\begin{equation}\label{Potentials}
\EF=-\mathbf{\nabla}\,\varphi\quad\quad\mbox{and}\quad\quad\BF=\mathbf{\nabla}\times\Apo.
\end{equation}
In the context of this work, the equation \eqref{EqExa} can be obtained by taking into account the following setup: $Q=\R^3$, $G=U(1)$ is the one-dimensional unitary group, $F$ is a coadjoint orbit of $G$ (consequently a point $-e\in\R$) with $\Phi$ the inclusion map. Needless to say, the connection $\Theta$ is determined locally by the vector potential $\Apo$; furthermore $L_{\sharp}=\frac{1}{2}\dot{\mathbf{q}}\cdot \dot{\mathbf{q}}-e\,\varphi(q)$. In this case, the first equation in \eqref{EoM}, i.e. $\dot z^{\alpha}=\Omega^{\alpha\beta}\lp\frac{\der L_{\sharp}}{\der z^{\beta}}-\bra\dot q^kA_k,\,\frac{\der\Phi}{\der z^{\beta}}\ket\rp$, leads to $0=0$, while the second reads 
\[
\ddot q^i=-e\,\der_i\varphi-e\,\dot q^j\lp\der_iA_j-\der_jA_i\rp,
\]
which according to \eqref{Potentials} is nothing but equation \eqref{EqExa}.

In the context of special relativity theory, both equations \eqref{EqEx} can be elegantly enclosed in the same condition by redefining the configuration manifold as the Minkowski space-time $Q=\R^{(1,3)}$, this is $\R^4$ endowed with a flat pseudo-Riemannian metric of Lorentz signature $(-,+,+,+)$. The rest of the setup remains the same, this is $G=U(1)$, $F=\{-e\}$ and $\Phi$ the inclusion map. In this new case, we establish the coordinates $q^{\mu}=(ct,q^i)$ for the configuration manifold (where $c$ is the speed of light), while the momentum  $T^*Q$ is determined locally by $p_{\mu}=(\En/c,p_i)$. We fix the connection by the local expression $A_{\mu}=(\varphi/c,A_i)$ (where the components are the potentials in \eqref{Potentials}) and the new Lagrangian function reads
\begin{equation}\label{LagEx}
L_{\sharp}=\frac{1}{2}\eta(w_q,w_q)=\frac{1}{2}\eta_{\mu\nu}\,\frac{dq^{\mu}}{d\tau}\,\frac{dq^{\nu}}{d\tau},
\end{equation}
where $\tau$ is re-scaling of the usual time $t$ by $c$ (in the following we will set $c=1$ for simplicity), $w_q=\frac{dq^{\mu}}{d\tau}\frac{\der}{\der q^{\mu}}\in T_qQ$ and $\eta:TQ\otimes TQ\Flder\R$ is the pseudo-Riemannian metric with local form $\eta(\der/\der q^{\mu},\der/\der q^{\nu})=\eta_{\mu\nu}=\mbox{diag}\,(-,+,+,+)$. To fix the notation, we shall denote $\frac{dq^{\mu}}{d\tau}=\dot q^{\mu}=(1,\dot q^i)$ and therefore  $L_{\sharp}=\frac{1}{2}\eta_{\mu\nu}\,\dot q^{\mu}\,\dot q^{\nu}$; besides $p_{\mu}=\eta_{\mu\nu}\dot q^{\nu}$ since the metric provides us with an isomorphism between $TQ$ and $T^*Q$. In this new setup the first equation in \eqref{EoM} is again $0=0$ while the second reads
\begin{equation}\label{momenta}
\frac{d}{d\tau}p_{\mu}=\,-e\,\dot q^{\nu}\lp\frac{\der A_{\mu}}{\der q^{\nu}}-\frac{\der A_{\nu}}{\der q^{\mu}}\rp.
\end{equation}
Recalling that $A_{\mu}$ is independent of time and that $q^0=t$, this equation may be decomposed as
\[
\begin{split}
\dot p_{i}=&\,-e\,\dot q^{\nu}\lp\frac{\der A_{i}}{\der q^{\nu}}-\frac{\der A_{\nu}}{\der q^{i}}\rp=\,e\,\dot q^{0}\lp\frac{\der A_{i}}{\der q^{0}}-\frac{\der A_{0}}{\der q^{i}}\rp-e\,\dot q^{j}\lp\frac{\der A_{i}}{\der q^{j}}-\frac{\der A_{j}}{\der q^{i}}\rp\\
=&\,-e\,\dot q^{0}\frac{\der \varphi}{\der q^{i}}-e\,\dot q^{j}\lp\frac{\der A_{i}}{\der q^{j}}-\frac{\der A_{j}}{\der q^{i}}\rp,\\\\
\dot p_{0}=&\,-e\,\dot q^{\nu}\lp\frac{\der A_{0}}{\der q^{\nu}}-\frac{\der A_{\nu}}{\der q^{0}}\rp=\,e\,\dot q^{0}\frac{\der A_{0}}{\der q^{0}}-\,e\,\dot q^{i}\frac{\der A_{0}}{\der q^{i}}=\,-e\,\dot q^{i}\frac{\der\varphi}{\der q^{i}},
\end{split}
\]
from which, considering that $p_0=\En$ and $p_i=\delta_{ij}\dot q^j$ and taking into account equation \eqref{Potentials}, we recover the equations \eqref{EqEx}, i.e.
\[
\delta_{ij}\ddot q^j=e\,\delta_{ij}E^j+e\,\epsilon_{ijk}\dot q^jB^k\quad\quad\mbox{and}\quad\quad\dot{\mathcal{E}}=e\,\dot{\mathbf{q}}\cdot\EF,
\]
where $\epsilon_{ijk}$ is the Levi-Civita tensor.
\medskip 

Now, we employ the approach developed in this work to reobtain these equations. First, consider the Lagrangian function \eqref{LagEx}, which in the space $\Fcalshd$ and its coordinates $(q^{\mu},v^{\mu},z^{\alpha})$ is redefined by $L_{\sharp}=\frac{1}{2}\eta_{\mu\nu}v^{\mu}v^{\nu}$. Therefore, the equations of motion obtained from the Hamilton-Sternberg-Pontryagin principle \eqref{eqconp} read in this case $\frac{dq^{\mu}}{d\tau}=\dot q^{\mu}=v^{\mu}$, $p_{\mu}=\frac{\der L_{\sharp}}{\der v^{\mu}}=\eta_{\mu\nu}v^{\nu}$ and $\dot p_{\mu}=\,-e\,\dot q^{\nu}\lp\frac{\der A_{\mu}}{\der q^{\nu}}-\frac{\der A_{\nu}}{\der q^{\mu}}\rp$; thus we recover \eqref{momenta}. On the other hand, regarding the Dirac structure $D^{\sharp}$ and the Hamilton-Dirac system $(D^{\sharp},E_{L_{\sharp}})$, the generalized Energy $E_{L_{\sharp}}:\PSB\Flder\R$ \eqref{GenEner} reads
\[
E_{L_{\sharp}}=\bra p,v\ket-L_{\sharp}(q,v,z)=p_{\mu}v^{\mu}-\frac{1}{2}\eta_{\mu\nu}v^{\mu}v^{\nu}.
\]
Taking into account the two-form $\Omsh$, the equations of motion of the Pontryagin-Sternberg Hamilton-Dirac system defined in proposition \ref{Lag-Dirac} are written as
\[
\lp
\begin{array}{ccc}
\dot q^{\nu}&
\dot v^{\nu}&
\dot p_{\nu}
\end{array}
\rp\lp
\begin{array}{ccc}
e(\der_{\nu}A_{\mu}-\der_{\mu}A_{\nu}) & 0&\delta_{\nu}^{\mu}\\
0&0&0\\
-\delta_{\mu}^{\nu}& 0&0
\end{array}
\rp\,=
\lp
\begin{array}{ccc}
0 &
p_{\mu}-\eta_{\mu\nu}v^{\nu}&
v^{\mu}
\end{array}
\rp,
\]
which after a straightforward computation leads to \eqref{momenta}.

\section{Conclusions}

In this paper, we have explored the construction of Hamilton-Dirac structures in the defined Pontryagin-Sternberg bundle, which we show is the suitable space to obtain, from different points of view, the equations of motion for charged particles in gauge fields. We apply the theory to a charged particle coupled to an electromagnetic field, field represented by a connection in a $U(1)$ principal bundle. However, our setting is general enough to cover also non-abelian groups. Our beginning point is the symplectic Sternberg phase space $(\Fcalsh,\Omsh)$, upon which we have constructed an analogue of the Pontryagin bundle $TQ\oplus T^*Q$, that we have named the Sternberg-Pontryagin bundle $\Fcalsh\oplus\Fcalshd$. We have related this bundle to the magnetized Tulczyjew triple \cite{Meng} analogously to how the Pontryagin bundle is related to the usual Tulczyjew triple. Then, we have shown that this is the suitable space to derive the equations of motion of particles in gauge fields from variational and intrinsic points of view (in the Lagrangian side). Moreover, we also show that it is necessary to define a (degenerate) extended Lagrangian function when deriving the equations in these contexts, extended Lagrangian which is highly influenced by the geometry of the Sternberg phase space. On the other hand, we have employed the Dirac structures theory to induce a Hamilton-Dirac system on $\Fcalsh\oplus\Fcalshd$ whose dynamical equations are equivalent to the equations under study. We have proved that this Dirac space generates naturally the desired dynamics and, furthermore, the needed Lagrangian function (which can be also degenerate) is simpler than the extended one proposed previously, i.e. it does not need to be {\it extended}.

\medskip

{\bf Acknowledgements:} I would like to thank Hiroaki Yoshimura for introducing me to Hamilton-Dirac systems, and Carlos Navarrete-Benlloch for reading part of this manuscript. I also thank the referee for valuable comments and corrections.


\begin{thebibliography}{99}


 \let\\, \newcommand{\by}[1]{\textsc{\ignorespaces #1}\\}
  \newcommand{\vol}[1]{{\bf{\ignorespaces #1}}}
  \newcommand{\info}[1]{\textrm{\ignorespaces #1}.}


\bibitem{Bai}
\by{Bai Z, Meng G and Wang E}
\title{``On the orbits of magnetized Kepler problems in dimension $2k+1$'',}
\info{Journal of Geometry and Physics, {\bf 73}, pp. 260--269, (2013)}

\bibitem{Yo0}\by{Barbero-Li\~n\'an M, Farr\'e Puiggal\'i M and Mart\'in de Diego D}\title{``Isotropic submanifolds and the inverse problem for mechanical constrained systems'',} \info{ Preprint, \href{http://arxiv.org/abs/1404.1961}{arXiv:1404.1961}, (2014)}

\bibitem{Yo1/2}\by{Barbero-Li\~n\'an M, de Le\'on M and Mart\'in de Diego D}\title{``Lagrangian submanifolds and Hamilton-Jacobi equation'',} \info{Monatschefte f\"ur Mathematik, {\bf 171(1-3)}, pp. 269--290, (2013)}


\bibitem{Cedric}
\by{Campos CM, Guzm\'an E and Marrero JC}
\title{``Classical field theories of first order and Lagrangian submanifolds of premultisymplectic manifolds,''}
\info{J. Geom. Mech. {\bf 4(1)}, pp. 1--26, (2012)}

\bibitem{CLMM2003}
\by{Cort\'es J, de Le\'on M,
Mart\'in de Diego D and Mart\'inez S}  \title{``Geometric description of
vakonomic and nonholonomic dynamics. Comparison of solutions'',}
\info{SIAM J. Control Optim., \textbf{41(5)}, pp. 1389--1412, (2003)}



\bibitem{Courant} \by{Courant TJ} \title{``Dirac manifolds,''} \info{Trans. Amer. Math. Soc. {\bf 319(2)}, pp. 631--661, (1990)}

\bibitem{CoWe1988} \by{Courant TJ and Weinstein A} \title{``Beyond Poisson structures`''} \info{Action hamiltoniennes
de groupes. Troisieme theor\'eme de Lie (Lyon, 1986), volume 27 of Travaux en Cours, pp.
39--49, (1988)}

\bibitem{Edu}
\by{Garc\'ia-Tora\~no E, Guzm\'an E, Marrero JC and Mestdag T}
\title{``Reduced dynamics and Lagrangian submanifolds of symplectic manifolds''}
\info{J. Phys. A, {\bf 47(22)}, 24pp., (2014) }

\bibitem{Yo1}\by{de Le\'on M, Jim\'enez, F and Mart\'in de Diego D}\title{``Hamiltonian dynamics and constrained variational calculus: continuous and discrete settings'',}\info{Journal of Physics A, {\bf 45}, 29 pp., (2012)}


\bibitem{Manolo}\by{de Le\'on N and Rodrigues PR}\title{``Methods of Differential Geometry in Analytical
Mechanics'',}\info{North-Holland, Amsterdam (1989)}

\bibitem{Dirac1}\by{Dirac PAM}\title{``Generalized Hamiltonian dynamics'',} \info{Canadian J. Math., {\bf 2}, pp. 129--148, (1950)}

\bibitem{Dirac2}\by{Dirac PAM}\title{``Lectures on Quantum Mechanics'',} \info{Belfer Graduate School of Science,
Yeshiva University, New York, (1964)}


\bibitem{Godbillon}
\by{Godbillon} \title{``G\'eometrie diff\'erentielle et m\'ecanique analytique'',} \info{Hermann, Paris (1969)}


\bibitem{Gra}\by{Grabowska K and Grabowski J}\title{``Variational calculus with constraints on general
algebroids'',} \info{Journal of Physics A, {\bf 41}, (2008)}

\bibitem{GG}\by{Grabowska K, Grabowski J and Urba\'nski} \title{``Geometrical Mechanics on algebroids'',}
\info{J. Geom. Meth. Mod. Physics., \vol{3}, pp. 559–-575, (2006)}

\bibitem{JiYo} \by{Jim\'enez F and Yoshimura H}\title{``Dirac Structures in Vakonomic Mechanics'',} \info{Accepted by Journal of Geometry and Physics, \href{http://arxiv.org/abs/1405.5394}{arXiv:1405.5394}, (2014)}




\bibitem{MMTulczy1995}
\by{Mendella M, Marmo M and Tulczyjew WM} 
\title{``Integrability of implicit differential equations''} \info{Journal of Physics A: Mathematical and General, {\bf 28(1)}, pp. 149--164, (1995)}

\bibitem{Meng0}\by{Meng G}\title{``The Poisson realization of $so(2,2k+2)$ on magnetic leaves and generalized MICZ-Kepler problems'',} \info{Journal of Mathematical Physics, {\bf 54}, 052902, (2013)}

\bibitem{Meng1}
\by{Meng G}
\title{``The classical magnetized Kepler problems in higher odd dimensions'',}
\info{J. Geomm. Symmetry Physics, {\bf 32}, pp. 15--32, (2013)}

\bibitem{Meng} \by{Meng G} \title{``Tulczyjew's approach for particles in gauge fields'',} \info{Preprint, \href{http://arxiv-web3.library.cornell.edu/abs/1405.0748v1}{arXiv:1405.0748}, (2014)}


\bibitem{Mont}\by{Montgomery R}\title{``Canonical formulation of a classical particle in a Yang-Mills field and Wong's equations'',} \info{Lett. Math. Phys. {\bf 8}, pp. 59--67, (1984)}


\bibitem{Pradines}
\by{Pradines J} \title{``Fibr\'es vectoriels doubles et calcul des jets non holonomes'',} \info{Amiens (1974)}


\bibitem{Waco}
\by{Robinson M, Bland K, Cleaver G and Dittmann J}
\title{``A Simple Introduction to Particle Physics'',}
\info{Review, \href{http://arxiv.org/abs/0810.3328}{arXiv:0810.3328}, (2009)}

\bibitem{shaft}
\by{van der Schaft AJ and Maschke BM} \title{``The Hamiltonian formulation of
energy conserving physical systems with external ports'',} \info{Archiv f\"ur Elektronik und
\"Ubertragungstechnik, {\bf 49}, pp. 362--371, (1995)}



\bibitem{Sni}\by{Sniatycki J}\title{``Geometric Quantization and Quantum Mechanics'',}\info{University of Calgary, Alberta, (1977)} 


\bibitem{Stenberg}\by{Stenberg S}\title{``Minimal coupling and the symplectic mechanics of a classical
particle in the presence of a Yang-Mills field'',} \info{Proc. Nat. Acad. Sci. {\bf 74}, pp. 5253--5254, (1977)}

\bibitem{So}\by{Souriau JM}\title{``Structure des Systemes Dynamiques'',} \info{Dunod, Paris, (1970)}

\bibitem{Tulczy3}\by{Tulczyjew WM}\title{``Hamiltonian systems, Lagrangian systems, and the Legendre transformation'',} \info{Symposia
Mathematica {\bf vol XIV} (Convegno di Geometria Simplettica e Fisica Matematica, INDAM, Rome), pp. 247--258, (1973)}


\bibitem{Tulczy1}
\by{Tulczyjew WM}\title{ ``Les sous-vari\'et\'es lagrangiennes et la dynamique hamiltonienne'',} \info{C. R. Acad. Sc. Paris {\bf 283} S\'erie A, pp. 15--18, (1976)}
%
\bibitem{Tulczy2}
\by{Tulczyjew WM} \title{``Les sous-vari\'et\'es lagrangiennes et la dynamique lagrangienne'',} \info{C. R. Acad. Sc. Paris {\bf 283} S\'erie A, pp. 675--678, (1976)}



\bibitem{Weinstein}\by{Weinstein A} \title{``A universal phase space for particles in Yang-Mills fields'',} \info{Lett. Math. Phys. {\bf 2}, pp. 417--420, (1978)}

\bibitem{Weins0}\by{Weinstein A}\title{``Symplectic manifolds and their Lagrangian submanifolds'',} \info{Advances in Mathematics, {\bf 6(3)}, pp. 329--346, (1971)}

\bibitem{Weins1}\by{Weinstein A}\title{``Lectures on symplectic manifolds'',}\info{CBMS Regional Conference Series
in Mathematics, {\bf 29}. American Mathematical Society, Providence, R.I., (1979)}


\bibitem{Wong}\by{Wong SK}\title{``Field and particle equations for the classical Yang-Mills fields and particles with isotropic spin'',} \info{Il Nuovo Cimento A, {\bf 65(4)}, pp. 689--694, (1970)}


\bibitem{YaMi}
\by{Yang CN and Mills RL}
\title{``Conservation of Isotropic Spin and Isotropic Gauge Invariance'',}
\info{Physical Review, {\bf 96(1)}, pp. 191--195, (1954)}


\bibitem{YM}\by{Yoshimura H and Marsden JE}\title{``Dirac structures in Lagrangian mechanics Part I: Implicit Lagrangian systems'',}\info{Journal of Geometry and Physics, {\bf 57}, pp. 133--156, (2006)}

\bibitem{YM2}\by{Yoshimura H and Marsden JE}\title{``Dirac structures in Lagrangian mechanics Part II: Variational structures'',}\info{Journal of Geometry and Physics, {\bf 57}, pp. 209--250, (2006)}

%

\end{thebibliography}
\end{document}